%% file: main.tex
\newif\ifECdraft
\ECdraftfalse

\ifECdraft
\documentclass[format=acmsmall, review=false]{acmart}

\usepackage{acm-ec-17}

\usepackage{booktabs} 

\usepackage[ruled]{algorithm2e} 

\SetAlFnt{\small}
\SetAlCapFnt{\small}
\SetAlCapNameFnt{\small}
\SetAlCapHSkip{0pt}
\IncMargin{-\parindent}

\newtheorem{claim}{Claim}

\else
\documentclass[11pt]{article}
\usepackage{fullpage} \usepackage{amssymb,amsmath,amsthm,amsfonts}
\usepackage[pdftex,final,bookmarks=true,pdfstartview=FitH,colorlinks,linkcolor=blue,filecolor=blue,citecolor=blue,urlcolor=blue]{hyperref}
\usepackage{algorithmicx}
\usepackage{algorithm}

\newtheorem{theorem}{Theorem}
\newtheorem{definition}{Definition}[section]
\newtheorem{corollary}[theorem]{Corollary}
\newtheorem{lemma}{Lemma}
\newtheorem{claim}{Claim}

\fi

\input{headers}

\begin{document}
\title{Does Nash Envy Immunity?}  

\ifECdraft
\author{Ching-Hua Yu}
	\affiliation{%
		\institution{University of Illinois, Urbana-Champaign}
		\department{Department of Computer Science}
		\city{Urbana}
		\state{IL}
		\postcode{61801}
		\country{USA}
	}

\input{abs}
\maketitle


\else
\author{Ching-Hua Yu
	\thanks{Computer Science Department, University of Illinois, Urbana-Champaign, cyu17@illinois.edu.}
}

\maketitle

\input{abs}
\thispagestyle{empty}
\newpage

 \thispagestyle{empty}
 \tableofcontents
 \thispagestyle{empty}
 \pagestyle{empty}
 \newpage

\pagestyle{plain}
\fi


\input{intro}

\input{prelims}
\input{envyfree}

\ifECdraft
	\begin{acks} 
		The author would like to thank Ruta Metha for valuable comments and discussion. 
	\end{acks}
	\bibliographystyle{ACM-Reference-Format}
	\bibliography{main}
\else
	\section*{Acknowledgement} The author would like to thank Ruta Metha for valuable comments and discussion. 	

\input{main.bbl}
\fi

\end{document}

%% file: headers.tex
\usepackage[T1]{fontenc}
\usepackage{amssymb,amsmath,amsfonts}
\usepackage{mathtools}
\usepackage{textcomp}
\usepackage[utf8]{inputenc}

\usepackage{xspace}
\usepackage[mathscr]{eucal}
\usepackage{color}
\usepackage{graphicx}

\usepackage{tikz}
\usetikzlibrary{arrows,positioning,calc,shapes.symbols,shapes.callouts,patterns}

\usepackage{nicefrac}
\usepackage{bm}
\usepackage{mathtools}

\usepackage{algpseudocode}

\usepackage{footnote}
\makesavenoteenv{tabular}
\makesavenoteenv{table}

\usepackage{ifdraft}
\ifdraft{
\newcommand{\NOTE}[1]{\textcolor[rgb]{0.1,0.6,.1}{NOTE: {#1}}}
\newcommand{\TODO}[1]{\textcolor[rgb]{1,0.5,.1}{{#1}}}

\newcommand{\marginTODO}[1]{%
  \marginpar{%
    \framebox[0.9\marginparwidth][t]{%
      \parbox[t]{0.85\marginparwidth}{%
        \raggedright\scriptsize{#1}%
   }}}%
}
}
{
\newcommand{\TODO}[1]{}
\newcommand{\NOTE}[1]{}
\newcommand{\marginTODO}[1]{}
}

\usepackage{times}
\newcommand{\uglyparagraph}[1]{\smallskip\noindent{\bf{#1}}}
\renewcommand{\paragraph}{\uglyparagraph}

\usepackage{paralist}
\usepackage{enumitem}
\setlist[itemize]{leftmargin=*}
\setlist[enumerate]{leftmargin=*}

\newenvironment{proofsketch}{\noindent{{\em Proof sketch:~}}}{\hfill $\Box$ \smallbreak}

\newcommand{\namedref}[3]{\hyperref[#2:#3]{#1~\ref*{#2:#3}}}

\newcommand{\sectionref}[1]{\namedref{Section}{sec}{#1}}

\newcommand{\theoremref}[1]{\namedref{Theorem}{thm}{#1}}
\newcommand{\definitionref}[1]{\namedref{Definition}{def}{#1}}

\newcommand{\lemmaref}[1]{\namedref{Lemma}{lem}{#1}}

\newcommand{\figureref}[1]{\namedref{Figure}{fig}{#1}}
\newcommand{\tableref}[1]{\namedref{Table}{tab}{#1}}

\newcommand{\NN}{\ensuremath{\mathbb N}\xspace}

\mathtoolsset{centercolon}

\newcommand{\LitSet}{{\ensuremath{{\text{LitSet}}}}\xspace}

\renewcommand{\P}{\ensuremath{\mathcal P}\xspace}

\tikzset{%
auto,
 pnode/.style = {rectangle,draw,fill=white,rounded corners=5pt,minimum size=1cm},
 bnode/.style = {rectangle,draw,ultra thick,gray,fill=white,minimum size=1cm}
}

\pgfdeclarelayer{bg}    
\pgfsetlayers{bg,main}  

%% file: abs.tex
\begin{abstract}
\small

The most popular stability notion in games should be Nash equilibrium under the rationality of players who maximize their own payoff individually. In contrast, in many scenarios, players can be (partly) irrational with some unpredictable factors. Hence a strategy profile can be more robust if it is resilient against certain irrational behaviors. In this paper, we propose
a stability notion that is resilient against envy. 
A strategy profile is said to be \emph{envy-proof} if by deviation, each player cannot gain a competitive edge with respect to the change in utility over the other players.
Together with Nash equilibrium and another stability notion called immunity, we show how these separate notions are related to each other, whether they exist in games, and whether and when a strategy profile satisfying these notions can be efficiently found. We answer these questions by starting with the general two player game and extend the discussion for the approximate stability and for the corresponding fault-tolerance notions in multi-player games.

 
 

\normalsize
\end{abstract}

%% file: intro.tex


\section{Introduction}



We study an inherent yet unpredictable factor of players in games that is related to the notorious envy\footnote{In English, malicious envy is more like jealousy, but here we adopt a general term \textit{envy} from the philosophy.} of humanity.
It is observed that in a company, the (un)happiness of employees can depend on, instead of their own salary,
the comparison between their salary and their colleagues'. In some scenarios, the situation can be worse. For example, a (somehow) irrational bidder in an auction can set a bid higher than the true value of an item just because he is unhappy to see his rival bidder win the item. 

In general, the outcome of games can deviate from the predicted equilibria, and the players' behavior can disobey their presumed utility function. Even though one remedy is to redefine the utility function and make it as accurate as possible, sometimes the utility function of a real player is hard to predict or to accurately model, and due to the complex humanity, sometimes faulty players can just become spontaneous or unwilling to follow a reasonable utility under a model. Hence it is attractive if an equilibrium has more robust properties against these inherent factors and can prevent potential irrational behavior of players.


To capture the inherent irrational factor due to envy, we define a property called \emph{envy-proofness}.
An envious player would not be satisfied with her current state if she discovers that she can gain more utility "than others", or she would lose less utility "than others" by changing her strategy.
An envy-proof strategy profile is stable with respect to this agitation. That is, A player cannot gain a competitive edge with respect to the utility change over the other players. Note that an envy-proof state may not be \emph{envy-free}: A player in an envy-proof state can envy the other's utility compared to her own, but she has no incentive to deviate.
Envy-proofness is independent of Nash equilibria in general; a Nash equilibrium can be envy-proof or not. 
Let's look at its definition for two-player games first.



\begin{definition}[Envy-proof Profile]
	In a two-player game $G=(\{u_0,u_1\},\{A_0,A_1\})$, we say a strategy profile $\mathbf{x}= (\mathbf{x_0}, \mathbf{x_1})$, a random variable over $A_0\times A_1$, is envy-proof if and only if the following holds:	
	\begin{align}\label{eq:envyfree0}
	&\notag \forall b, \forall x_b'\in A_b,
	u_b(\mathbf{x_b'}:\mathbf{x_{\bar{b}}})-u_b(\mathbf{x})\leq 	 u_{\bar{b}}(\mathbf{x_b'}:\mathbf{x_{\bar{b}}})-u_{\bar{b}}(\mathbf{x}), \\
	&\notag (\text{where the notation } u_b(\mathbf{x})=E_{x\sim \mathbf{x}}[u_b(x)]), \\
	&\notag \text{abbreviated as} \\
	& \ \ \ \ \ \ \ \ \ \ \ \ \ \ \ \ \ \ \ \ \ \ 
	\forall b, \triangle_{b}u_b(\mathbf{x})\leq \triangle_{b}u_{\bar{b}}(\mathbf{x})
	\end{align}
\end{definition}

\begin{figure}[ht] \large
	\centering
	\begin{tabular}{|l|l|l|}
		\hline
		& movie & shopping \\ \hline
		movie    & 4,4   & 1,3      \\ \hline
		shopping & 3,1   & 3,3      \\ \hline
	\end{tabular}
	\caption{A coordination game. (movie, movie) is a Nash equilibrium but not envy-proof, while (shopping, shopping) is an envy-proof Nash equilibrium.}
	\label{fig:envyfree}
\end{figure}

As an example, consider the coordination game described in \figureref{envyfree}. The two players like shopping no mater whether they do it alone or together. On the other hand, they only like to go to the movies together. While (movie, movie) and (shopping, shopping) both sound like a good coordinated strategy profile, one day, the two players quarrel with each other, behave irrationally, and even try to make each other unhappy regardless self-happiness. Then (movie, movie) is no longer a good strategy since it is unstable. It is unstable because each player has an incentive to change the strategy, and assuming the other's strategy is fixed, this deviation can cause the other player to lose more happiness (losing 3 points vs losing 1 point). On the other hand, (shopping, shopping) remains a good strategy because any unilateral deviation of a player can only cause his own loss.  


In addition to envy-proofness, we study another 
stability notion called immunity. 
%
In \cite{GradwohlRe08,GradwohlRe14}, Gradwohl and Reingold proposed and studied  $(\epsilon,t)$-immunity for special games in a fault-tolerance genre. Essentially a strategy profile is $(\epsilon,t)$-immune if
any $t$ players cannot deviate and make some other player's expected utility decrease more than $\epsilon$ (in absolute value).
They showed that in some special games (which we will review briefly later), every Nash equilibrium is $(\epsilon,t)$-immune. 
While we will adapt the envy-proof notion to fault-tolerance issues in multi-layer games later in \sectionref{multiplayer},
we focus more on the fundamental properties such as the existence and the computational complexity of exact as well as approximate solutions in general games. Adapting the notion of immunity to two player games, we say a strategy profile is immune if each player's unilateral deviation cannot decrease the other player's utility.

\begin{definition}[Immunity Profile]
	In a two-player game $G=(\{u_0,u_1\},\{A_0,A_1\})$, we say a strategy profile $\mathbf{x}$ is immune if and only if:
	\begin{equation*}\label{eq:twoImmune0}
	\forall b, \forall x_b'\in A_b, u_{\bar{b}}(x_b':\mathbf{x_{\bar{b}}})\geq u_{\bar{b}}(\mathbf{x}),
	\end{equation*}
\end{definition}

This paper is then centered around the three stability notions: Nash equilibria, envy-proofness, and immunity and developed from multiple dimensions. First we start by inspecting general two-player games and illustrate a clean connection between the three stability notions in \sectionref{connection}. And in \sectionref{existence}, we study the (non)existence of these notions. We show that while every two-player game has at least an immune strategy profile and an envy-proof strategy profile, an immune Nash equilibrium and an envy-proof Nash equilibrium do not always exist.
Then following these observations regarding existence, in \sectionref{solvability} we are concerned about whether they can
be efficiently found in a game. We show that while finding an immune profile is \emph{PPAD-complete}, finding an envy-proof strategy profile is polynomial-time solvable.

Based on the above discoveries, we check advanced properties of envy-proof Nash equilibria and immune Nash equilibria. In \sectionref{EnvyImmune}, somehow surprisingly, we show while determining the existence of envy-proof Nash equilibria is \emph{NP-complete}, determining the existence of immune Nash equilibria is polynomial-time solvable. Next, because of the PPAD-hardness of finding an immune strategy profile and the NP-hardness of determining the existence of an envy-proof Nash equilibrium, in \sectionref{approx}, we relax the definition of immunity and envy-proofness to resort to  approximate solutions. Finally in \sectionref{multiplayer}, we discuss the multi-player fault-tolerance version of these notions. Particularly we show possitive results of $(\epsilon,t)$-coalitional envy-proofness for $\gamma$-varied games,
$\gamma$-sensitive games, and anonymous games.



\section{Related Works and Discussion}\label{sec:relatedwork}
\subsection{Terminology Disambiguation for Envy-Related Notions}

Envy-related notions appeared a few times in several areas. In 1958, Gamow and Stern \cite{GamowSt58} introduce an envy-free concept for fair cake-cutting and chore division.  
Essentially the envy-free concept means that each partner believes that their share is at least as great as any other share. Note that the envy-freeness in the problem is quite different from our envy-proofness for general games. Particularly instead of comparing the utilities of the players, we want to compare the change in utilities of the players. That is, the utility between players can be substantially different but a strategy profile would be envy-proof if by any unilateral deviation of a player, the increase of her utility is less or the decrease of her utility is more than the others'. 



The term envy-freeness and local envy-freeness also appeared in the problem of generalized second-price (GSP) auction, given by Varian \cite{Varian07} and Edelman et al. \cite{EdelmanOsSc07} independently.
The concept there is to capture the utility optimization of rational players (while excludes some plausible equilibria that cannot provide a good utility). In contrast, the envy-proofness in this paper is motivated by irrational players. Particularly an envy-free bid profile in GSP is necessarily an equilibrium, while an envy-proof profile in a game is not necessary a Nash equilibrium. For example, an irrational bidder in an auction can deviate to decrease his payoff if this can decrease his rival bidder' payoff more.

\subsection{Irrational Players in Games}

While the rationality of the players is a primal assumption in game theory, recently variant points of view are proposed to study or to explain irrational player behaviors in the computer science literature in recent years. 

\paragraph{Fault Tolerance.} One of the closest directions to our work is about fault tolerance of Nash equilibria against certain kinds of malicious attacks in large games. Following the line, Bernheim et al. \cite{BernheimBeWh87}, Kalai \cite{Kalai04}, Gradwohl and Reingold \cite{GradwohlRe10,GradwohlRe14} showed that Nash equilibria are robust against certain kind of attacks in general games satisfying specific properties. Particularly, Gradwohl and Reingold \cite{GradwohlRe14} showed that an honest player's payoff in a Nash equilibrium remains almost the same when the number of Byzantine faulty players is less than a threshold. The kind of tolerance is called immunity. 

An instant question is that even though approximate immune Nash equilibria exist in games with a special form, it is unclear whether exact/approximate immune profiles or immune Nash equilibria exist in general games and whether it is computationally hard to find a solution. We answered these questions in this work, and so in some sense, we provide a fundamental inspections for the complement points following the line. Also, we study a new notion called envy-proofness that enriches the player's behaviors between fully rational and fully malicious. Interestingly, we will show Nash equilibria, immunity, and envy-proofness are different but relative notions in these questions.

\paragraph{Spitefulness.} Another related direction is about the spiteful/altruistic setting. The setting says that a player's "true" utility is a combination of his own and the others. 
(See Chen\cite{PAChen11thesis} and Chen and Micali{ChenMicali16} for detailed reviews and references on this line.)
One major difference between their direction and ours is that the spiteful/altruistic setting assumes certain knowledge of the external factors (called externalities) that decides the true utility function. 
This is a necessary step in their approaches since their goals are about how the profit of an auction problem or the social welfare of a game would compromise due to these externalities and how an (approximate) solution of an optimization problem according to the redefined utilities can be obtained.

In contrast, we consider scenarios whether the external factors are unpredictable and focus on issues about Nash equilibria tolerant to irrational players. That is, we assume the metric of the given utility functions are accepted by all the parties and then study equilibria robust against (arbitrarily) malicious or envious behaviors of the players. Nevertheless, we believe that the robustness of Nash equilibria studied in this paper can potentially extend to related problems in their setting. Also, while we show general impossibility and computational hardness results of some of our stability notions, it would be interesting to inspect on more specific problems, as those in this line, for the future studies.


Also, the \emph{price of malice}, defined as the ratio between the social welfare when there are malicious players and the social welfare when there is none, was studies by Moscibroda et al. \cite{MoscibrodaScWa10} in a virus inoculation game and by Babaioff et al. \cite{MalicePapadimitriou07} and Roth \cite{MaliceRoth08} in congestion games. However their works are incomparable to ours since we leave out questions about social welfare for now and focus on questions about equilibria themselves.

\paragraph{Rational Cryptography.}  In contrast to a standard assumption where each player is either honest in following the protocol or arbitrarily malicious, rational cryptography concerns scenarios involving both rational and Byzantine players. See \cite{RationalVadhan09,RationalMicali14,RationalMicali13,RationalMicali12,RationalKatz12}, to name a few, and surveys by Nielsen \cite{SurveyNielsen07} and Katz \cite{SurveyKatz09}. On one hand, this setting helps the implementation of protocols go beyond the impossibility regarding fairness in the standard setting. On the other hand, it brings new problems in protocol designs that characterize the complex features of player behaviors. 

The tolerance threshold of Byzantine players remains an important issue in secure multiparty computation with rational/partial-rational players, and the design of such protocols usually involve the equilibrium properties (instead of social welfare). Therefore the fundamental studies of the existence, complexity, and approximation of the strategy profiles with respect to the stability notions in the paper has a potential for future development in secure protocols under a rich composition of players.

%% file: prelims.tex


\section{Definition}

\begin{definition}[Game]
	A normal-form game is described by $G=(\{u_i\}_{i\in[m]},\{A_i\}_{i\in[m]})$ as follows:
	\begin{itemize}
		\item $m$ indicates that there are $m$ players $\P_1,...,\P_m$.
		\item $u_i$ is the utility function of $\P_i$.
		\item $A_i$ is a finite set of $\P_i$'s  actions.
	\end{itemize}
	
\end{definition}

When $A_i$'s are all the same, we also denote the game by $G=(\{u_i\}_{i\in[m]},A)$. We use a bold font $\mathbf{x}$ to denote a random variable and w.l.o.g. a distribution over a given space. 
We will slightly abuse the notation for the expected value of a utility function $u$ by writing $u(\mathbf{x})=E_{x\sim \mathbf{x}}[u(x)]$.

\begin{definition}[Strategy Profile] For a game $G=(\{u_i\}_{i\in[m]},\{A_i\}_{i\in[m]})$, a pure strategy profile is a vector $x=(x_1,...,x_m)$ in $\bigotimes_{j=1}^m A_j$, and a mixed strategy profile (or profile for short) is a random variable $\mathbf{x}=(\mathbf{x_1},...,\mathbf{x_m})$ over $\bigotimes_{j=1}^m A_j$.
\end{definition}

We will also denote a strategy profile by $(\mathbf{x}_{-i}:\mathbf{x_i'})$, where $\mathbf{x}_{-i}$ is a product distribution over $\bigotimes_{j\neq i}^m A_j$, and $(\mathbf{x}_{-i}:\mathbf{x_i'})
=(\mathbf{x_i'}:\mathbf{x}_{-i})=(\mathbf{x}_1,...,\mathbf{x_{i-1}},\mathbf{x_{i}'},\mathbf{x_{i+1}}...,\mathbf{x_m})$

\begin{definition}[Nash Equilibrium]
	Let $G=(\{u_i\}_{i\in[m]},\{A_i\}_{i\in[m]})$ be a game. We say a strategy profile $\mathbf{x}$ is a (mixed) Nash equilibrium if and only if:
	\begin{equation*}\label{eq:twoNash}
	\forall i, \forall x_i'\in A_i, u_i(x_i':\mathbf{x_{-i}})\leq u_i(\mathbf{x}), 
	\end{equation*}
\end{definition}

For two-payer games, we will denote the players by $\P_0$, $\P_1$ and the game by $G=(\{u_0,u_1\},\{A_0,A_1\})$. Similarly, a profile $(\mathbf{x_0},\mathbf{x_1})$ indicates that $\P_0$ plays $\mathbf{x_0}$, and $\P_1$ plays $\mathbf{x_1}$, while a profile
$(\mathbf{x_b}:\mathbf{x_{\bar{b}}})$, the same as $(\mathbf{x_{\bar{b}}}:\mathbf{x_b})$, means $\P_b$ plays $\mathbf{x_b}$ and $\P_{\bar{b}}$ plays $\mathbf{x_{\bar{b}}}$ for some $b\in\{0,1\}$.


%% file: envyfree.tex
\section{Warm-up for Immunity and Envy-proofness}\label{sec:connection}

We first discuss the connection between the three notions: Nash Equilibrium, immunity, and envy-proofness in two-player games. Interestingly, their connection is very clean. 
Let's formalize and compare the definitions of these notions in general two-player games:

\begin{definition}[Immunity for 2-Player Games\footnote{We note that if we adapt the formula of $(\epsilon,t)$-immunity by Gradwohl and Reingold \cite{GradwohlRe14} to two-player games without the relaxation $\epsilon$, it would be $u_{\bar{b}}(x_b':\mathbf{x_{\bar{b}}})\leq u_{\bar{b}}(\mathbf{x})$ instead. However, we insist on our form that captures the idea that a Byzantine player has no incentive to deviate from an immune profile if he cannot decrease his rival player's payoff by doing so.}]
	In a two-player game $G=(\{u_0,u_1\},\{A_0,A_1\})$, we say a strategy profile $\mathbf{x}$ is immune if and only if:
	\begin{equation*}\label{eq:twoImmune}
	\forall b, \forall x_b'\in A_b, u_{\bar{b}}(x_b':\mathbf{x_{\bar{b}}})\geq u_{\bar{b}}(\mathbf{x}).
	\end{equation*}	
	
\end{definition}

\begin{definition}[Envy-proofness for 2-Player Games]
	In a two-player game $G=(\{u_0,u_1\},\{A_0,A_1\})$, we say a strategy profile $\mathbf{x}$ is envy-proof if and only if	
	\begin{align}\label{eq:envyfree}
	&\notag \forall b, \forall x_b'\in A_b,
	u_b(\mathbf{x_b'}:\mathbf{x_{\bar{b}}})-u_b(\mathbf{x})\leq 	 u_{\bar{b}}(\mathbf{x_b'}:\mathbf{x_{\bar{b}}})-u_{\bar{b}}(\mathbf{x}), \\
	&\text{abbreviated as \ } \forall b, \triangle_{b}u_b(\mathbf{x})\leq \triangle_{b}u_{\bar{b}}(\mathbf{x}).
	\end{align}
\end{definition}

\paragraph{Compare envy-proofness with immunity.} Immunity is a stability notion against a complete irrational adversary who can be arbitrarily malicious. She would take any action if it can cause a loss of the other player. On the contrary, envy-proofness is against a partial rational player who is envious about the other player's payoff and compare their payoffs all the time. We strengthen that the goal of envy-proofness is not to divide a cake equally or to make everyone happy (which is a too strong requirement in many real cases). Instead, envy-proofness captures the notion that each player has no incentive to change her mind even if she is an envious player.

\paragraph{Compare envy-proofness with Nash equilibria.} While a Nash equilibrium says that a rational player $\P_b$ has no incentive to deviate if $u_b(x_b':\mathbf{x_{\bar{b}}})\leq u_b(\mathbf{x})$, an envy-proof profile suggests that an envious player has no incentive to deviate if $u_b(x_b':\mathbf{x_{\bar{b}}})-u_{\bar{b}}(x_b':\mathbf{x_{\bar{b}}})\leq u_b(\mathbf{x})-u_{\bar{b}}(\mathbf{x})$, which replaces the utility function $u_b(\cdot)$ in the Nash with the difference of the utilities, $(u_b-u_b')(\cdot)$.

\begin{theorem}\label{thm:NashEnvyImmune}
 For a 2-player game $G=(\{u_0,u_1\},\{A_0,A_1\})$, 	an immune Nash equilibrium is envy-proof.
\end{theorem}

\begin{proof}
Let $\mathbf{x}$ be an immune Nash equilibrium. $\forall b, \forall x_b'$, a Nash equilibrium states
\begin{equation*}
u_b(x_b':\mathbf{x_{\bar{b}}})- u_b(\mathbf{x})\leq 0, 
\end{equation*}
and an immune profile suggests
\begin{equation*}
u_{\bar{b}}(x_b':\mathbf{x_{\bar{b}}})- u_{\bar{b}}(\mathbf{x})\geq 0,
\end{equation*}
which together imply
\begin{equation*}
\triangle_{b}u_b(\mathbf{x})\leq 0 \leq \triangle_{b}u_{\bar{b}}(\mathbf{x}).
\end{equation*}
\end{proof}

\section{On (Non)-Existence}\label{sec:existence}

We know that every finite game admits at least one Nash equilibrium \cite{Nash50}, but how about the immune profile and the envy-proof profile? Also, does every game admit an immune Nash equilibrium or at least an envy-proof Nash equilibrium? Here we answer these questions for two-player games. Particularly, we show that while every two-player game admits at least one Nash equilibrium, one immune profile, and one envy-proof profile, they are not necessarily to cover each other. That is, some games do not admit any envy-proof Nash equilibrium.

	
\begin{lemma}\label{lem:immune}
	Every finite two-player game admits an immune profile.
\end{lemma}
\begin{proof}
Let $G=(\{u_0,u_1\},\{A_0,A_1\})$ be a 2-player game. An immune profile $\mathbf{x}$ requires 
\begin{equation*}
\forall b, \forall x_b', u_{\bar{b}}(x_b':\mathbf{x_{\bar{b}}})\geq u_{\bar{b}}(\mathbf{x}),
\end{equation*}
Define a game $G'=(\{-u_1,-u_0\},\{A_0,A_1\})$. A Nash equilibrium $\mathbf{y}$ of $G'$ satisfies
\begin{equation*}
\forall b, \forall y_b', -u_{\bar{b}}(y_b':\mathbf{y_{\bar{b}}})\leq -u_{\bar{b}}(\mathbf{y}).
\end{equation*}
Hence $\mathbf{y}$ must be an immune profile of $G$. 
Since there is always a Nash equilibrium in $G'$, there is also an immune profile in $G$.
\end{proof}	

\begin{lemma}\label{lem:envyfree}
	Every finite two-player game admits an envy-proof profile.
\end{lemma}
\begin{proof}
	Let $G=(\{u_0,u_1\},\{A_0,A_1\})$ be a 2-player game. From \eqref{eq:envyfree}, an envy-proof profile $\mathbf{x}$ says
	\begin{equation*}
		\forall b, \triangle_{b}u_b(\mathbf{x})\leq \triangle_{b}u_{\bar{b}}(\mathbf{x})
	\end{equation*}
	Define $G'=(\{u_0-u_1,u_1-u_0\},\{A_0,A_1\})$. A Nash equilibrium $\mathbf{y}$ of $G'$ says
	\begin{equation*}
	\forall b, \forall y_b', u_b(y_b':\mathbf{y_{\bar{b}}})-u_{\bar{b}}(y_b':\mathbf{y_{\bar{b}}})\leq u_b(\mathbf{y})-u_{\bar{b}}(\mathbf{y}).
	\end{equation*}
	Hence $\mathbf{y}$ must be an envy-proof profile of $G$. 
	Since there is always a Nash equilibrium in $G'$, there is also an envy-proof profile in $G$.
\end{proof}

\begin{figure*}[!tp] \large
	\centering
	\begin{tabular}{|l|l|l|}
		\hline
		& movie & shopping \\ \hline
		movie    & 4,4   & 1,3      \\ \hline
		shopping & 3,1   & 0,0      \\ \hline
	\end{tabular}
	\caption{A coordination game. (movie, movie) is the only Nash equilibrium while (shopping, shopping) is the only envy-proof strategy profile.}
	\label{fig:nonexistenvy}
\end{figure*}

\begin{lemma}\label{lem:NoEnvyFreeNash}
	There are two-player games which do not admit any envy-proof Nash equilibrium.
\end{lemma}
\begin{proof}
	Consider a two-player coordination game as described in \figureref{nonexistenvy}. Note that the game admits only one Nash equilibrium (movie, movie) and only one envy-proof strategy profile (shopping, shopping). Hence there is no envy-proof Nash equilibrium in the game. 
\end{proof}

\begin{corollary}\label{cor:NoImmuneNash}
	There are two-player games which do not admit any immune Nash equilibrium.	
\end{corollary}
\begin{proof}
	This is from \theoremref{NashEnvyImmune} and \lemmaref{NoEnvyFreeNash}.
\end{proof}

\section{On Efficient Solvability}\label{sec:solvability}
After the discussion of (non)-existence, we are curious about whether these notions are efficiently solvable. Particularly for finding a Nash equilibrium, we already know the problem is PPAD-complete, but how about finding an immune profile and find an envy-proof profile? 

\begin{lemma}\label{lem:immuneHard}
	Finding an immune profile in a finite two-player game is PPAD-complete.
\end{lemma}
\begin{proofsketch}
	This is similar to the proof of \lemmaref{immune} but argued in an inverse way.
	Consider any 2-player game $G=(\{u_0,u_1\},\{A_0,A_1\})$. Define a game $G'=(\{-u_1,-u_0\},\{A_0,A_1\})$. Then the problem of finding a Nash equilibrium of $G$ is reduced to finding an immune profile of $G'$. Since finding Nash is PPAD-complete, it is PPAD-complete, too.
\end{proofsketch}

\begin{lemma}\label{lem:envyfreeEasy}
	Finding an envy-proof profile in a finite two-player game is polynomial-time solvable.
\end{lemma}
\begin{proofsketch}
	In light of the proof of \lemmaref{envyfree}, finding an envy-proof profile can reduce to finding a Nash equilibrium of a zero-sum game. Since zero-sum two-player game is polynomial-time solvable, so is finding an envy-proof profile.
\end{proofsketch}

Recall that from \theoremref{NashEnvyImmune},  envy-proofness is a potentially weaker notion than immunity through the reduction sense. Here by \lemmaref{immuneHard} and \lemmaref{envyfreeEasy}, it is clear to see this point.

\section{Envy-proof Nash and Immune Nash}\label{sec:EnvyImmune}

Now we learned that finding a Nash profile and finding an immune profile are PPAD-hard; meanwhile find an envy-proof profile is efficiently solvable. However, how about immune Nash equilibria and envy-proof Nash equilibria? We have also illustrated that neither an immune Nash equilibrium nor an envy-proof Nash equilibrium does always exist in a game. Hence a follow-up question is whether we can determine their existence in a game efficiently, and whether we can find a solution efficiently if there does exist one.




Gilboa and Zemel \cite{GilboaZe89} and Conitze and Sandholm \cite{ConitzeSa08} show that 
it can be NP-complete to determine the existence of a Nash equilibrium with certain simple constraints.
Their results however cannot imply the problem here, nor does their construction in the proof work out. Nevertheless, following the line of reduction from SATISFIABILITY, we show that determining the existence of an envy-proof Nash equilibrium is NP-complete.

\begin{theorem}\label{thm:envyfreeNashHard}
	Even in symmetric 2-player games, it is NP-complete to determine whether there is an envy-proof Nash equilibrium.
\end{theorem}

We defer the proof to \sectionref{envyfreeNashHard}. The idea is to construct a game such that every legal assignment of a Boolean formula corresponds to an envy-proof profile in the game, and every Nash equilibrium, except a default Nash equilibrium, corresponds to a satisfiable assignment of the formula. Meanwhile, the default, always-existing Nash equilibrium serves as an absorbing state in the game but is not envy-proof.


On the other hand, somehow surprisingly, while finding an immune profile is harder than finding an envy-proof profile (by \lemmaref{immuneHard} and \lemmaref{envyfreeEasy}), the following theorem suggests that determining the existence of an immune Nash equilibrium is easier than determining the existence of an envy-proof Nash equilibrium, and immune Nash equilibria can be efficiently located given they exist in a game.

\begin{theorem}\label{thm:immuneNashEasy}
	Let $G=(\{u_0,u_1\},\{A_0,A_1\})$ be a normal form 2-player game. It is polynomial-time solvable to determine whether there exists an immune Nash equilibrium in $G$ and to output a solution if there exists.
\end{theorem}

\begin{proof}
	By definition, a strategy profile $\mathbf{x}$ is an immune Nash equilibrium if and only if
	
	\begin{equation}\label{eq:ImmuneNash}
		\forall b, \forall x_b'\in A_b,x'_{\bar{b}}\in A_{\bar{b}},
		u_b(\mathbf{x_b}:x_{\bar{b}}')\geq u_b(\mathbf{x})\geq u_b(x_b',\mathbf{x_{\bar{b}}}).
	\end{equation}
	
	Now define two games $G_0, G_1$: For all $b, G_b=(\{u_b,-u_b\},A_b)$. Note $\mathbf{x}$ is a Nash equilibrium of $G_0$ and $G_1$ if and only if $\mathbf{x}$ satisfies \eqref{eq:ImmuneNash}. Hence $\mathbf{x}$ is an immune Nash equilibrium of $G$ if and only if $\mathbf{x}$ is a Nash equilibrium of $G_0$ and $G_1$ at the same time. Because both $G_0$ and $G_1$ are zero-sum games, 
	finding their Nash equilibria equals solving linear programming $LP_0$ and $LP_1$: For  $b=0$ and $1$, let $M_b$ be the matrix form of $u_b$. Then,
	
	\begin{align}
	 LP_b:  & \max z_b\\
	 \notag s.t. \ \ & \mathbf{x}^T M_b\geq z_b \mathbf{1}^T  \\
	 &\notag \mathbf{x}^T\mathbf{1} = 1 \\
	 &\notag \mathbf{x} \geq \mathbf{0}
	\end{align}
	
	Clearly $G$ admits an immune Nash equilibrium if and only if $LP_0$ and $LP_1$ share a common optimum solution. To check whether they do, we solve $LP_0$ and $LP_1$ first. Let $v_b=\max z_b, \forall b$. We claim that the common optimum solution space of $LP_0$ and $LP_1$ is nonempty if and only if the following inequations are feasible:
	
	\begin{equation}\label{eq:immuneNashLP}
	\forall b, \mathbf{x}^T M_b\geq v_b \mathbf{1}^T
	\end{equation}
	
	Indeed, if $\mathbf{x}$ is a common optimum solution of $LP_0$ and $LP_1$, then certainly it fulfills \eqref{eq:immuneNashLP}. On the other hand, if $\mathbf{x}$ fulfills \eqref{eq:immuneNashLP}, it implies a feasible solution to $LP_b$ with value at least $v_b$ for each $b$. As $v_b$ is the optimum value of $LP_b$, $\mathbf{x}$ is then an optimum solution to $LP_b$ for each $b$.
	Since linear programming is polynomial-time solvable, so is the problem.
	
\end{proof}

\input{app-satproof}

\section{On Approximation}\label{sec:approx}
Here we study the additive approximation of the three stability notions. For example, while an immune profile says that each player's unilateral deviation cannot decrease the rival player's utility, an $\epsilon$-immune profile says that each player's unilateral deviation cannot decrease the rival player's utility more than $\epsilon$ (in absolute value). Intuitively, an approximate immune profile is stable against a lazy Byzantine player who attacks only when he can make sufficiently huge loss of his opponent.

\begin{definition} \label{def:approx}
	Let $G=(\{u_0,u_1\},\{A_0,A_1\})$ be a  2-player game and $\mathbf{x}$ be a strategy profile.  
	\begin{itemize}
		\item $x$ is an $\epsilon$-Nash equilibrium if $\forall b, \forall x_b'$, $u_b(x_b':\mathbf{x_{\bar{b}}})\leq u_b(\mathbf{x})+\epsilon$.
		\item $x$ is $\epsilon$-Immune if  $\forall b, \forall x_b'$, $u_{\bar{b}}(x_b':\mathbf{x_{\bar{b}}})\geq u_{\bar{b}}(\mathbf{x})-\epsilon$.
		\item $x$ is $\epsilon$-Envy-proof if  $\forall b$, $\triangle_{b}u_b(\mathbf{x})\leq \triangle_{b}u_{\bar{b}}(\mathbf{x})+\epsilon$.
	\end{itemize}
\end{definition}


From \lemmaref{immuneHard}, we learned that finding an immune profile of a two-player game $G$ is PPAD-complete. In fact, the proof of \lemmaref{immuneHard} indicates that finding an immune profile of $G$ is the same as finding a Nash equilibrium of another game $G'$ which switches the negative utilities of the two players. Thus, applying a sub-exponential algorithm for finding an $\epsilon$-Nash equilibrium by Lipton, Markakis and Mehta \cite{LiptonMaMe03}, we immediately get the following lemma.

\begin{lemma}\label{lem:approximmune}
	For a two-player game $G=(\{u_0,u_1\},A)$, where $|A|=n$, there is a $n^{O(\frac{\ln n}{\epsilon})}$ algorithm to find an $\epsilon$-Immune profile.
\end{lemma}

The algorithm from Lipton et al. \cite{LiptonMaMe03} is based on exhaustively searching all possible $k$-\textit{uniform} $\epsilon$-Nash equilibria, where "$k$-uniform" means that each strategy is a uniform distribution over a multiset of size $k$. They show, through the probability method, that there always exists a $k$-uniform $\epsilon$-Nash equilibria, where the players' utility is close to an exact Nash equilibria so the algorithm works.

Now let's look at a strategy profile with two stability characters again. 
Note that by Definition \ref{def:approx}, we say a strategy profile is $\epsilon$-envy-proof Nash equilibrium if it is a Nash equilibrium that is $\epsilon$-envy-proof. Similarly, a strategy profile is an $\epsilon_1$-envy-proof $\epsilon_2$-Nash equilibrium if it is an $\epsilon_2$ Nash equilibrium that is $\epsilon_1$-envy-proof. 
From \theoremref{envyfreeNashHard}, we learned that finding an envy-proof Nash equilibrium of a two-player game is NP-complete. So similarly, we resort for potential approximate solutions. 
However the theorem from \cite{LiptonMaMe03} does not apply here.
The reason is that even if we can find an $\epsilon$-Nash equilibrium and even if the players' utility can also be close to some exact envy-proof Nash equilibrium, it is unnecessary that the $\epsilon$-Nash equilibrium remains envy-proof.
Nevertheless by adapting the probability method to the problem of approximate envy-proofness here, it can be proved that there is a sub-exponential algorithm for either outputting  $\epsilon$-envy-proof $\epsilon$-Nash equilibrium or deciding non-existence of any envy-proof Nash equilibrium in a game.

In the following theorem, we show a more generalized result that has a $2\epsilon$-additive approximation gap between deciding the existence and deciding the non-existence of an approximate envy-proof Nash equilibrium in a game. 



\begin{theorem}\label{thm:apxapxefnash}
	Let $G=(\{u_0,u_1\},A)$ be a two-player game with $|A|=n$, $\epsilon>0$, and $k=\frac{3\ln n}{\epsilon^2}$.
	\begin{enumerate}
		\item There is a $O(n^{3k})$-time algorithm for computing a $\epsilon'$-envy-proof $\epsilon$-Nash equilibrium, where $\epsilon'$ is the minimum approximation factor of envy-proofness over all $k$-uniform $\epsilon$-Nash equilibria.
		\item If $\epsilon'>2\epsilon$, then there is no $(\epsilon'-2\epsilon)$-envy-proof Nash equilibrium in $G$.
	\end{enumerate}
	
\end{theorem}

\input{app-apxproof}

\section{Multiplayer Games}\label{sec:multiplayer}
We leave several extended questions about multiplayer games for the future, but before the end of this article, we would like to give some more points of view about fault tolerance against envious players in multi-player games. We will define coalitional envy-proofness as a natural extension and relaxation of envy-proofness in multi-player games and define a $\gamma$-varied game that captures the smooth change of the derivative utilities between players. 



\subsection{Preliminaries for Multiplayer Games}
First let us reminisce two fault tolerance notions in multiplayer games given in \cite{GradwohlRe14} that generalize approximate Nash equilibria and approximate immune profiles respectively.

\begin{definition}[$(\epsilon,t)$-immunity\cite{GradwohlRe14}]\label{def:etimmune} 
	In a $m$-player game $G=(\{u_i\}_{i\in[m]},A)$, a strategy profile $\mathbf{x}=\mathbf{x_1}\times\cdots\times\mathbf{x_m}$ is $(\epsilon,t)$-immune if for every set $S\subset N$ of size at most $t$, every $x_S'\in A^{|S|}$, and every $j\not\in S$,
	\begin{equation*}
	u_j(\mathbf{x}_{-S}:x_S')\geq u_j(\mathbf{x})-\epsilon
	\end{equation*}
\end{definition}

\begin{definition}[$(\epsilon,t)$-coalitional Nash equilibrium\cite{GradwohlRe14}] \label{def:etNash}
	In a $m$-player game $G=(\{u_i\}_{i\in[m]},A)$, a strategy profile $\mathbf{x}=\mathbf{x_1}\times\cdots\times\mathbf{x_m}$ is $(\epsilon,t)$-coalitional Nash equilibrium if for every set $S\subset N$ of size at most $t$, every $x_S'\in A^{|S|}$, and every $i\in S$,
	\begin{equation*}
	u_i(\mathbf{x}_{-S}:x_S')\leq u_i(\mathbf{x})+\epsilon
	\end{equation*}
\end{definition}

According to \cite{GradwohlRe14}, a profile is said to be $(\epsilon,t)$-immune "if players’ expected utilities do not decrease by more than $\epsilon$ when
any $t$ other players deviate arbitrarily." Nevertheless the formula in their paper is written as $u_i(\mathbf{x})\geq u_i(\mathbf{x}_{-S}:x_S')-\epsilon$. We checked that their main results would hold even if the inequality were revised to $u_i(\mathbf{x}_{-S}:x_S')\geq u_i(\mathbf{x})-\epsilon$. Since the later one is more natural and matches the described notion, we adopt it in \definitionref{etimmune}.

On the other hand, the notion of coalition has a long history in game theory, and similar notions include strong Nash equilibria \cite{Aumann59}, coalition-proof Nash equilibria \cite{BernheimBeWh87} and coalition-proof correlated strategies \cite{MorenoWo96}, etc. Compared to these notions, however, as in \cite{GradwohlRe14}, here we focus on the notion of fault tolerance with a bound on the number of corrupted players.

In \cite{GradwohlRe08,GradwohlRe14}, one main result is about the robustness of Nash equilibria in a $\lambda$-continuous game. We note that the game is equivalent to an alternatively stated $\gamma$-sensitive game in \cite{KearnsPaRoUl13}) for $\gamma=\lambda/m$, where $m$ is the number of players. Essentially each player's utility is quite insensitive to another player's change in the game, which is a relaxation of the insensitive game.

\begin{definition}[$\gamma$-sensitive game\cite{KearnsPaRoUl13,GradwohlRe08,GradwohlRe14}] 
	A game $G=(\{u_i\}_{i\in[m]},A)$ is said to be $\gamma$-sensitive if 
	\begin{equation*}
	\forall i, \forall x\in A^{m}, \forall x'_i\in A, \forall j\neq i, |u_j(x_{-i}:x_i)-u_j(x)|\leq \gamma.
	\end{equation*}
\end{definition}

From \cite{GradwohlRe14}, it was observed that $\gamma$-sensitive games are strongly fault-tolerant with respective to immunity and coalition as the following theorem.

\begin{theorem}[\cite{GradwohlRe14}]\label{thm:sensitiveImmune}
	Let $G$ be an $m$-player $\gamma$-sensitive game. Then every Nash equilibirum of $G$ is $(\gamma t,t)$-immune and $(3\gamma t,t+1)$-coalitional.
\end{theorem}

\subsection{Coalitional Envy-proofness}

Now we extend the definition of envy-proofness for the corresponding fault-tolerant notion in multi-player games, which we call $(\epsilon,t)$-coalitional envy-proofness as follows.


\begin{definition}[$(\epsilon,t)$-coalitional envy-proofness] \label{def:etenvyfree}
	In a $m$-player game $G=(\{u_i\}_{i\in[m]},A)$, a strategy profile $\mathbf{x}=\mathbf{x_1}\times\cdots\times\mathbf{x_m}$ is $(\epsilon,t)$-coalitional envy-proof if for every set $S\subset N$ of size at most $t$, every $x_S'\in A^{|S|}$, every $i\in S$, and every $j\not\in S$,
	\begin{equation*}
	u_i(\mathbf{x}_{-S}:x_S')- u_i(\mathbf{x}) \leq u_j(\mathbf{x}_{-S}:x_S')- u_j(\mathbf{x}) +\epsilon
	\end{equation*}
\end{definition}

In addition, we define a $\gamma$-varied game that is a natural relaxation of a (generalized) team gameas follows.


\begin{definition}[$\gamma$-varied game]
	A game $G=(\{u_i\}_{i\in[m]},A)$ is said to be $\gamma$-varied if 
	\begin{equation*}
	\forall i,j,k, x, |\triangle_i u_j(x)-\triangle_i u_k(x)|\leq \gamma.
	\end{equation*}
\end{definition}

Then we have an analog observation for fault-tolerance with respective to envy-proofness in $\gamma$-varied games as follows.

\begin{theorem}\label{thm:variedEnvyfree}
	Let $G$ be an $m$-player $\gamma$-varied game. Then every Nash equilibrium of $G$ is $(\gamma t,t)$-coalitional envy-proof.
\end{theorem}

\begin{proofsketch}
	Let $\mathbf{x}$ be a Nash equilibrium of $G$. For every $S\subset N, |S|\leq t$, let $S\equiv\{k_1,...,k_t\}$ and  $S_{\ell}\equiv\{k_1,...,k_\ell\}, \forall \ell\in [t]$. Then 	
	for every $x_S'\in A^{|S|}$, and every $i\in S, j\not\in S$
	\begin{align*}
	&  \left( u_i(\mathbf{x}_S:x_S')- u_i(\mathbf{x}) \right) - \left( u_j(\mathbf{x}_S:x_S')- u_j(\mathbf{x})\right)\\ 
	= &  \left( u_i(\mathbf{x}_{-k_1}:x_{k_1}')-u_i(\mathbf{x}) + \sum_{\ell=2}^t  u_i(\mathbf{x}_{-S_{\ell}}:x'_{S_{\ell}})-u_i(\mathbf{x}_{-S_{\ell-1}}:x'_{S_{\ell-1}}) \right)\\
	- &  \left( u_j(\mathbf{x}_{-k_1}:x_{k_1}')-u_j(\mathbf{x}) + \sum_{\ell=2}^t  u_j(\mathbf{x}_{-S_{\ell}}:x'_{S_{\ell}})-u_j(\mathbf{x}_{-S_{\ell-1}}:x'_{S_{\ell-1}}) \right)\\
	\leq & \gamma t
	\end{align*}
	
\end{proofsketch}

Next, we generalize \theoremref{NashEnvyImmune} for the corresponding fault-tolerant notions in multi-player games.

\begin{theorem}\label{thm:eNashEnvyImmune}
	Let $\mathbf{x}$ be a Nash equilibrium in a game $G$. If $\mathbf{x}$ is $(\epsilon_1,t)$-immune and $(\epsilon_2,t)$-coalitional, then $\mathbf{x}$ is $(\epsilon_1+\epsilon_2,t)$-coalitional envy-proof.
\end{theorem}

\begin{proofsketch} For every $S\subset N$ of size $|S|\leq t$, every $x_S'\in A^{|S|}$, and every $i\in S, j\not\in S$,
by $(\epsilon_2,t)$-coalition,
\begin{equation*}
u_i(\mathbf{x}_{-S}:x'_S)\leq u_i(x)+\epsilon_2,
\end{equation*}
\noindent and by $(\epsilon_1,t)$-immunity,
\begin{equation*}
u_j(\mathbf{x}_{-S}:x'_S)\geq u_j(x)+\epsilon_1.
\end{equation*}
Thus, 
\begin{equation*}
u_i(\mathbf{x}_{-S}:x'_S)- u_i(x) \leq 
u_j(\mathbf{x}_{-S}:x'_S)- u_j(x)+\epsilon_1+\epsilon_2.
\end{equation*}

\end{proofsketch}

Then by \theoremref{eNashEnvyImmune} and \theoremref{sensitiveImmune}, we have the following corollary.

\begin{corollary}
	Let $G$ be a $\gamma$-sensitive game. Then every Nash equilibirum of $G$ is $(4\gamma t,t)$-coalitional envy-proof.
\end{corollary}

Another corollary is about the existence of coalitional envy-proof Nash equilibria in anonymous games, in which the utility of each player is a function of his own action and the empirical distribution of the other players’ actions. Formally, a game $G=(\{u_i\}_{i\in[m]},A)$ is anonymous if for all $i$, for all permutation $\sigma:[m]\backslash \{i\}\rightarrow [m]\backslash \{i\}$, for all $x\in A^m$,  $u_i(x)=u_i(x_{\sigma(1)},...,x_{\sigma(i-1)},x_i,x_{\sigma(i+1)},...,x_{\sigma(m)})$.
In \cite{GradwohlRe14}, Gradwohl and Reingold showed that if an $m$-player anonymous game has a constant action space for each player, then the game has an $\epsilon$-Nash equilibrium that is $(\epsilon,t)$-immune and $(4\epsilon,t+1)$-coalitional for $t=O(\sqrt{m})$. Together with \theoremref{eNashEnvyImmune}, we have the following observation.

\begin{corollary}
	For every constant $d\in\NN,\epsilon\in (0,1)$, there exists a constant $c$ such that for every anonymous game $G=(\{u_i\}_{i\in[m]},A)$ with $|A|=d$, for $t=c\sqrt{m}$, $G$ admits a $\epsilon$-Nash equilibrium which is $(5\epsilon,t)$-coalitional envy-proof.
\end{corollary}



%% file: app-satproof.tex
\subsection{Proof of \theoremref{envyfreeNashHard}}\label{sec:envyfreeNashHard}

Let $\phi$ be a given Boolean formula in conjunctive normal form. Let $V=\{x^1,...,x^n\}$ be the set of variables, $L=\{+x^1,-x^1,...,+x^n,-x^n\}$ the set of corresponding literals, and $C$ the set of clauses in $\phi$. Define a function $\LitSet: V\cup C \rightarrow 2^{L}$ that outputs the set of literals corresponding to a variable or a clause. That is, $\LitSet(x)=\{+x,-x\}$ if $x\in V$ and $\LitSet(x)=\{\ell| x \text{ contains } \ell\}$ if $x\in C$. We construct a symmetric 2-player game $G(\phi)$ in normal form as follows. Let $f$ be a default action, $\Sigma=V\cup L\cup C\cup \{f\}$, $\beta=-n$ (or any sufficiently small value), and the utility function be

\begin{itemize}
	\item $u_0(\ell_0,\ell_1)=u_1(\ell_0,\ell_1)= \left\{\begin{matrix}
	n-1,	& \ell_0\neq -\ell_1 \\ 
	\beta,	& \ell_0= -\ell_1
	\end{matrix}\right. , \forall \ell_0,\ell_1\in L$
	\item $u_0(x_0,x_1)=u_1(x_0,x_1)=\beta, \forall x_0,x_1\in V\cup C$.
	\item $u_0(x,\ell)=u_1(\ell,x)= \left\{\begin{matrix}
	n, & \ell\not\in \LitSet(x) \\ 
	0, & \ell \in \LitSet(x)
	\end{matrix}\right. ,\forall x\in V\cup C$
	\item $u_0(\ell,x)=u_1(x,\ell)= \left\{\begin{matrix}
	n-2, & \ell\not\in \LitSet(x) \\ 
	2(n-1), & \ell \in \LitSet(x)
	\end{matrix}\right. ,\forall x\in V\cup C$
	\item $u_0(f,x)=u_1(x,f)=
	\left\{\begin{matrix}
	n+2, & x=f\\ 
	n, & x\in V\cup C\\ 
	n-1, & x\in L
	\end{matrix}\right.$
	\item $u_0(x,f)=u_1(f,x)=
	\left\{\begin{matrix}
	n+1, & x\in V\cup C\\ 
	n-1, & x\in L\\
	\end{matrix}\right.$	
\end{itemize}

Take $\phi=(x^1\vee -x^2)\wedge (-x^1\vee x^2)$ for example. The game $G(\phi)$ is showed in \tableref{sat}. Note that $\phi$ has exactly two solutions, either assigning both variables to true or both to false 
, and the game has exactly 3 equilibria: 1) both player play $\{+x^1,+x^2\}$ uniformly randomly ; 2) both player play $\{-x^1,-x^2\}$ uniformly randomly; 3) both players play $f$. The first two are envy-proof while the last one is not. 

\begin{table}[]
	\centering
	\caption{The table shows the game $G(\phi)$, where $\phi=(x^1\vee -x^2)\wedge (-x^1\vee x^2)$. By $\times$ we denote $(\beta,\beta)$ since $\beta$ is just a sufficiently small value for both the players to avoid running into the state.}
	\label{tab:sat}
	\begin{tabular}{|l|l|l|l|l|l|l|l|l|l|}
		\hline
		& $x^1$    & $x^2$    & $+x^1$   & $-x^1$   & $+x^2$   & $-x^2$   & $(x^1\vee -x^2)$ & $(-x^1\vee x^2)$ & $f$ \\ \hline
		$x^1$            & $\times$ & $\times$ & 0,2      & 0,2      & 2,0      & 2,0      & $\times$         & $\times$         & 3,2 \\ \hline
		$x^2$            & $\times$ & $\times$ & 2,0      & 2,0      & 0,2      & 0,2      & $\times$         & $\times$         & 3,2 \\ \hline
		$+x^1$           & 2,0      & 0,2      & 1,1      & $\times$ & 1,1      & 1,1      & 2,0              & 0,2              & 1,1 \\ \hline
		$-x^1$           & 2,0      & 0,2      & $\times$ & 1,1      & 1,1      & 1,1      & 0,2              & 2,0              & 1,1 \\ \hline
		$+x^2$           & 0,2      & 2,0      & 1,1      & 1,1      & 1,1      & $\times$ & 0,2              & 2,0              & 1,1 \\ \hline
		$-x^2$           & 0,2      & 2,0      & 1,1      & 1,1      & $\times$ & 1,1      & 2,0              & 0,2              & 1,1 \\ \hline
		$(x^1\vee -x^2)$ & $\times$ & $\times$ & 0,2      & 2,0      & 2,0      & 0,2      & $\times$         & $\times$         & 3,2 \\ \hline
		$(-x^1\vee x^2)$ & $\times$ & $\times$ & 2,0      & 0,2      & 0,2      & 2,0      & $\times$         & $\times$         & 3,2 \\ \hline
		$f$              & 2,3      & 2,3      & 1,1      & 1,1      & 1,1      & 1,1      & 2,3              & 2,3              & 4,4 \\ \hline
	\end{tabular}
\end{table}

We will proof the following claim, which implies \theoremref{envyfreeNashHard}.

\begin{claim} 
If $(\ell^1,...,\ell^n)$ satisfies $\phi$, then there is an envy-proof Nash equilibrium of $G(\phi)$, where both players play $\{\ell^1,...,\ell^n\}$ uniformly randomly. The only other Nash equilibrium is for both players to play $f$ and is not envy-proof.
\end{claim}

\begin{proof}
	
Let $\mathbf{\mu}$ be the uniform distribution over $\{\ell^1,...,\ell^n\}$. 
First we show that the profile $(\mathbf{\mu},\mathbf{\mu})$ is indeed an envy-proof Nash equilibrium. Because $G(\phi)$ is a symmetric game, w.l.o.g., we fix player $\P_0$'s strategy to be $\mathbf{\mu}$ and consider changes in player $\P_1$'s strategy. Since $u_1(\ell^i,\ell^j)=n-1$ for all $i,j\in[n]$, playing any distribution $\mathbf{\nu}$ over $\{\ell^1,...,\ell^n\}$ gives the same expected utility $u_1(\mathbf{\mu},\mathbf{\nu})=n-1$\footnote{Recall that when $\mathbf{\mu},\mathbf{\nu}$ is a random variable over the action space, the notation $u_1(\mathbf{\mu},\mathbf{\nu})$ denotes the expected utility $E_{\mathbf{\mu},\mathbf{\nu}}[u_1(\mathbf{\mu},\mathbf{\nu})]$}. Playing the negation of some $\ell^i$ gives $u_1(\mathbf{\mu},-\ell^i)=\frac{1}{n}\cdot \beta+\frac{n-1}{n}\cdot n<n-1$, playing a variable or a clause $x$ in $V\cup C$ gives 
$u_1(\mathbf{\mu},x)=\frac{1}{n}\cdot 0+\frac{(n-1)}{n}\cdot n=n-1$, and playing $f$ gives $u_1(\mathbf{\mu},f)=n-1$. Hence $u_1(\mathbf{\mu},\mathbf{a})\leq u_1(\mathbf{\mu},\mathbf{\mu})$ for any $\mathbf{a}$; $(\mathbf{\mu},\mathbf{\mu})$ is a Nash equilibrium. On the other hand, note $u_0(\mathbf{\mu},\mathbf{\nu})=u_1(\mathbf{\mu},\mathbf{\nu})$, $u_0(\mathbf{\mu},-\ell^i)=u_1(\mathbf{\mu},-\ell^i)$,  $u_0(\mathbf{\mu},x)=\frac{1}{n}\cdot 2(n-1)+\frac{(n-1)}{n}\cdot (n-2)=n-1=u_1(\mathbf{\mu},x)$, and $u_0(\mathbf{\mu},f)=n-1=u_1(\mathbf{\mu},f)$, which together imply $\triangle_1 u_1(\mathbf{\mu},\mathbf{\mu})=\triangle_1 u_0(\mathbf{\mu},\mathbf{\mu})$. Hence $(\mathbf{\mu},\mathbf{\mu})$ is envy-proof. 

As for the other Nash equilibrium $(f,f)$, since $u_b(f,f)$ is strictly larger than any other utility value, clearly it is a Nash equilibrium. However it is not envy-proof because $-1=u_1(f,x)-u_1(f,f)>u_0(f,x)-u_0(f,f)=-2$ for all $x\in V\cup C$.

Next we argue that there is no other Nash equilibria. If $\P_1$ plays $f$, clearly the unique best response for $\P_0$ is to also play $f$ since $u_b(f,f)$ gives the unique maximum utility. On the other hand, if $\P_1$ plays a mixed strategy $\mathbf{x_1}$ over $V\cup C$, and $\P_0$ plays a mixed strategy $\mathbf{x_2}$ over $\Sigma-\{f\}$, the social utility $\bar{u}(\mathbf{x_1},\mathbf{x_2})$ is at most $2(n-1)$. If $\P_0$'s utility is smaller than $n$, it is better for $\P_0$ to play $f$ instead (since $u_0(f,x)=n$ for all $x\in V\cup C$). Otherwise, it implies $\P_1$'s utility is at most $n-2$, and it is better for $P_1$ to play $f$ instead (since $u_1(x,f)\geq n-1$ for all $x$). Hence the profile cannot be a Nash equilibrium. The same analysis holds if $\P_1$ plays a mixed strategy over $\Sigma$ with nonzero probability to play actions in $V\cup C$ and $\P_0$ plays any mixed strategy over $\Sigma$ for the following two reasons: 1) Because $u_0(f,\ell)=u_0(\ell,\ell)$ and $u_1(\ell,f)=u_1(\ell,\ell)$, playing $f$ instead would not decrease this part's portion of utility. 2) Because $(f,f)$ gives the maximum utility, even the mixed strategy of $\P_1$ involves nonzero probability of $f$, playing $f$ instead would not decrease this part's portion of utility either, and vice versa. Hence the support of a Nash equilibrium cannot involves actions in $V\cup C$.

So, we only need to check strategy profiles over $L\cup \{f\}$. However, if $\P_1$ puts nonzero probability on $f$, then again, playing $f$ is a strictly better strategy for $\P_0$ since $u_0(f,f)$ is the unique maximum utility and $u_0(f,\ell)=u_0(\ell,\ell)$. Hence except $(f,f)$ the support of any other Nash equilibrium can only involve actions in $L$.

The remaining argument to exclude any other Nash equilibrium over $L$ is similar to that of \cite{ConitzeSa08}. Note that in this case, the expected utility for each player is at most $n-1$. If for some $x\in V$, $\P_0$ puts a probability on playing either $+x$ or $-x$ less than $\frac{1}{n}$, then the expected utility for $\P_1$ to play $x$ is strictly greater than $\frac{1}{n}\cdot 0+\frac{n-1}{n}\cdot n= n-1$. 
Furthermore, if for some $\ell\in L$, $\P_0$ puts positive probability on $\ell$ while $\P_1$ on $-\ell$, their utility would be less than $n-1$, and switching to $f$ would be a strictly better strategy.
Hence, these cannot be Nash equilibrium, and we can assume that for each $x\in V$, exactly one of $\{+x,-x\}$ is played with probability $\frac{1}{n}$, and $\P_0$ and $\P_1$ choose the same literal set. Finally, if they play a set of literals whose corresponding assignment does not satisfy a clause $c$, then playing $c$ instead would give utility $n$ and would be a better strategy for both players.

\end{proof}

%% file: app-apxproof.tex
\subsection{Proof of \theoremref{apxapxefnash}}\label{sec:approxefnash}

We will first prove that if there exists an $\epsilon_{ef}$-envy-proof Nash equilibrium, then there exists a $k$-uniform $(2\epsilon+\epsilon_{ef})$-envy-proof $\epsilon$-Nash equilibrium. Then by an exhaustive search, we can check all $\binom{n+k-1}{k}^2$ possible $k$-uniform $\epsilon$-Nash equilibrium in time $O(n^{3k})$. Suppose after the exhaustive search, we cannot find any $k$-uniform $\epsilon$-Nash equilibrium that is $(2\epsilon+\epsilon_{ef})$-envy-proof, then by contradiction, there does not exist any $\epsilon_{ef}$-envy-proof Nash equilibrium in the game.

For simplicity we assume the utilities have values between 0 and 1, and both players have $n$ pure strategies. We note that the result can be generalized without these restrictions.

\begin{lemma}
	If a game $G=(\{u_0,u_1\},\{a_0^1,...,a_0^n\},\{a_1^1,...,a_1^n\})$ admits an $\epsilon_{ef}$-envy-proof Nash equilibrium $\mathbf{x^*}\equiv(\mathbf{x_0^*},\mathbf{x_1^*})$, then for every $\epsilon<1$, for every $k\geq \frac{3\ln n}{\epsilon^2}$, there exists a $k$-uniform $(2\epsilon+\epsilon_{ef})$-envy-proof $\epsilon$-Nash equilibrium.
\end{lemma}

\begin{proof}
	
	First we construct a profile $\hat{x}\equiv(\mathbf{\hat{x}_0},\mathbf{\hat{x}_1})$ by sampling from $(\mathbf{x_0^*},\mathbf{x_1^*})$: For $b=\in\{0,1\}$, form multiset $S_b$ by sampling $k$ pure strategies with replacement according to $\mathbf{x_b^*}$ independently. Let $\mathbf{\hat{x}_b}$ be a mixed strategy that assigns $1/k$ probability to each member in $S_b$.
	
	Second let $GOAL$ be the event when $\mathbf{\hat{x}}$ is a $(2\epsilon+\epsilon_{ef})$-envy-proof $\epsilon$-Nash equilibrium, and define events\footnote{Recall that when $\mathbf{x_{\bar{b}}^*}$ is a mixed strategy, our notation $u_b(a_b^i:\mathbf{x_{\bar{b}}^*})$ denotes the the expected utility over the randomness of $\mathbf{x_{\bar{b}}^*}$, and $(a_b^i:\mathbf{x_{\bar{b}}^*})$ denotes an unordered strategy profile where $\P_b$ plays $a_b^i$, and $\P_{\bar{b}}$ plays  $\mathbf{x_{\bar{b}}^*}$.} 
	\begin{align*}
		\pi_{b,i} & = \{u_b(a_b^i:\mathbf{x_{\bar{b}}^*})\leq u_b(\mathbf{x^*})+\epsilon\}, \forall b\in\{0,1\},i\in [n], and \\
		\theta_{b,i} & = \{u_b(a_b^i:\mathbf{x_{\bar{b}}^*})- u_b(\mathbf{x^*}) \leq u_{\bar{b}}(a_b^i:\mathbf{x_{\bar{b}}^*})- u_{\bar{b}}(\mathbf{x^*})+2\epsilon+\epsilon_{ef}\}, \forall b\in\{0,1\},i\in [n] \\
		& \text{so,	} GOAL = \bigcap_{b,i} (\pi_{b,i}\cap \theta_{b,i}).
	\end{align*}
	
	Third, in order to apply probability arguments, we need to somehow decompose $\pi_{b,i}$ and $\theta_{b,i}$. Define the following events:
	\begin{align*}
		\phi_b &=\{|u_b(\mathbf{\hat{x}})-u_b(\mathbf{x^*})|\leq \epsilon/2 \}, \forall b \\
		\psi_{b,i} &=\{|u_b(a_b^i:\mathbf{\hat{x}_{\bar{b}}})-u_b(a_b^i:\mathbf{x^*_{\bar{b}}})|\leq \epsilon/2 \}, \forall b	
	\end{align*}
		It can be verified that $\phi_b\cap \psi_{b,i} \subseteq  \pi_{b,i}, \forall b,i$, and $\bigcap_{b'} (\phi_{b'}\cap \psi_{b',i}) \subseteq \theta_{b,i}, \forall b,i$ 		
	
	To analyze the probability of event $\phi_b$, we further define	
		\begin{align*}
		\phi_{b,1} &=\{|u_b(\mathbf{\hat{x}_b}:\mathbf{x^*_{\bar{b}}})-u_b(\mathbf{x^*})|\leq \epsilon/4 \}, \forall b \\
		\phi_{b,2} &=\{|u_b(\mathbf{\hat{x}})-u_b(\mathbf{\hat{x}_b}:\mathbf{x^*_{\bar{b}}})|\leq \epsilon/4 \}, \forall b
		\end{align*}
	Clearly, $\phi_{b,1}\cap \phi_{b,2}\subseteq \phi_{b}$. Put them together: $\bigcap_{b}\phi_{b,1}\bigcap_b\phi_{b,2}\bigcap_{b,i}\psi_{b,i}\subseteq GOAL$
	
	Finally we bound the probability of the events. 
	Due to the construction of $\mathbf{\hat{x}_b}$, $u_b(\mathbf{\hat{x}_b}:\mathbf{x^*_{\bar{b}}})$ (resp. $u_b(\mathbf{\hat{x})}$) is like the average of $k$ independent random variables of the same expected values $u_b(\mathbf{x^*})$ (resp. $u_b(\mathbf{\hat{x}_b}:\mathbf{x^*_{\bar{b}}})$). Hence by Hoeffding bound \cite{Hoeffding63},
	\begin{equation*}
	\Pr[\phi_{b,j}^c]\leq 2e^{-k\epsilon^2/8}, \forall b\in\{0,1\},j\in\{1,2\}
	\end{equation*}
	
	Similarly, 
	\begin{equation*}
	\Pr[\psi_{b,i}^c]\leq 2e^{-k\epsilon^2/2}, \forall b\in\{0,1\},i\in[n]
	\end{equation*}
	
	Then by union bound, we have 
	\begin{equation*}
	\Pr[GOAL^c]\leq \sum_{b,j} \Pr[\phi_{b,j}^c]+\sum_{b,i} \Pr[\psi_{b,i}^c]\leq 8e^{-k\epsilon^2/8}+4ne^{-k\epsilon^2/2}<1
	\end{equation*}
	
	Since $GOAL$ happens with a nonzero probability, there exists a $k$-uniform $(2\epsilon+\epsilon_{ef})$-envy-proof $\epsilon$-Nash equilibrium.
	
\end{proof}